\documentclass[conference]{IEEEtran}

\usepackage{amssymb,amsmath,amsthm}
\usepackage{mathrsfs}
\usepackage[utf8]{inputenc}
\usepackage[T1]{fontenc}
\usepackage{graphicx}
\usepackage{epsfig}
\usepackage[english]{babel}
\usepackage{subfigure}
\usepackage{color}
\usepackage{import}
\usepackage{multirow}
\usepackage{cases}
\usepackage{mathtools}
\usepackage{color, colortbl}
\usepackage{cite}

\newtheorem{definition}{Definition}
\newtheorem{proposition}{Proposition}
\newtheorem{remark}{Remark}

\definecolor{LightCyan}{rgb}{0.85,1,1}

\usepackage{todonotes}

\usepackage[xcolor]{changebar}
\newcommand{\totalpower}[1]{%
	\overline{P}_{#1}
}



\begin{document}

\title{Fronthaul-Aware Software-Defined Joint Resource Allocation and User Scheduling  for 5G Networks}
\author{      \IEEEauthorblockN{Chen-Feng Liu, Sumudu Samarakoon, and Mehdi Bennis}
\IEEEauthorblockA{ Centre for Wireless Communications, University of Oulu,  Finland\\E-mail: \{cliu, sumudu, bennis\}@ee.oulu.fi}}

\maketitle
\begin{abstract}


Software-defined networking (SDN) is the concept of decoupling the control and data planes to create a flexible and agile network, assisted by a central controller. However, the performance of SDN highly depends on the limitations in the fronthaul which are inadequately discussed in the existing literature. In this paper, a fronthaul-aware software-defined resource allocation mechanism is proposed for 5G wireless networks with  in-band wireless fronthaul constraints. Considering the fronthaul capacity, the controller maximizes the time-averaged network throughput by enforcing a coarse correlated equilibrium (CCE) and incentivizing base stations (BSs) to locally optimize their decisions to ensure mobile users' (MUs) quality-of-service (QoS) requirements. By marrying tools from Lyapunov stochastic optimization and game theory, we propose a two-timescale approach where the controller gives recommendations, i.e., sub-carriers with low interference, in a long-timescale whereas BSs schedule their own MUs and allocate the available resources in every time slot. Numerical results show  considerable throughput enhancements and delay reductions over a non-SDN network baseline.

\end{abstract}


\section{Introduction}\label{Sec: Introduction}

To satisfy the ever-increasing traffic demands in wireless networks, small cell base stations (BSs) are deployed to boost network capacity. However, due to frequency reuse in adjacent cells, signals from neighboring BSs result in severe co-channel interference. Although coupled BSs can coordinate their transmissions by exchanging information \cite{Chen:14:Mag}, it incurs high overhead in terms of latency, power cost, and allocated bandwidth.
In order to reduce the overhead of coordination, the concept of software-defined networking (SDN) which decouples the control and data plane of networks has been considered \cite{2015:SOF,Gudipati:13:SDN,Le2015}. Under current wireless SDN architectures,  a central controller, having the global view of the network at the top of the hierarchy, makes control decisions and issues recommendations to BSs at the lower-level of the hierarchy.
Thus, BSs are allowed to locally optimize their transmissions following the controller's recommendations without further coordination among neighboring BSs.
In \cite{Fu:2013:NVF}, Fu {\it et al.}~consider a scenario in which wireless service providers (WSPs) bid for resources via a central controller on behalf of the subscribed users. The resource management is analyzed through an auction process. In  \cite{Xianfu:SDN},  a central  controller schedules users based on the WSP's value function which relies on other competitors' private information. Additionally, the flow programmability of SDN under dense mobile environments is studied in \cite{Kozat:15:Flow}.
However, the above studies do not address the overhead induced  by using a capacity-limited shared fronthaul which negatively impacts the network performance.
Hence, the impact of fronthaul capacity and latency on the network performance cannot be ignored, especially in highly-centralized SDN architectures.

The main contribution of this paper is to propose a fronthaul-aware SDN architecture for a network consisting of locally coupled small cell BSs. These BSs compete over limited resources to maximize their own throughputs while ensuring mobile users' (MUs) quality-of-service (QoS) requirements.
 Due to the time-varying environment and mutual interference coupling among BSs, the throughput maximization problem is modeled as a {\it stochastic game}. In the SDN architecture, the central controller optimizes the weighted sum-log utility subject to a game theoretic equilibrium known as coarse correlated equilibrium (CCE), and provides recommendations for BSs via an in-band wireless fronthaul.
 The reliability of the fronthaul is measured in terms of the signal-to-noise ratio (SNR) between the BS and the controller
%
while the cost of fronthaul is modeled in terms of a time penalty as a function of the number of BSs and MUs. To reduce the overhead in the fronthaul, we propose a two-timescale decomposition: optimization at the central controller in a long timescale and decision making at BSs during a short timescale.
Using tools from Lyapunov optimization, a low-complexity  resource allocation and MU scheduling scheme is employed at each BS.
In the numerical results, we evaluate the performance of the proposed SDN-aware resource allocation solution by comparing it to a non-SDN architecture. Furthermore, the impact of fronthaul reliability on the performance  is studied with the aid of the simulations.


\section{System Model}\label{Sec: System model}

We consider a downlink (DL) wireless network consisting of a set of locally coupled small cell BSs $\mathcal{B}$ with frequency reuse factor 1. In this network, BS $b$ serves a set of MUs $\mathcal{M}_b$, and each MU is associated with one BS only.
 A central controller coordinates these locally coupled BSs and communicates with BSs via an in-band wireless fronthaul. We assume that all network entities are equipped with single antenna and share a set of sub-carriers $\mathcal{S}$ in the DL.
Additionally,  the considered network operates in slotted time indexed by $t\in\mathbb{Z}^{+}$, and  each time slot is unit time for simplicity. In this regard, we refer to $T_0$ slots as one time frame $\mathcal{F}$ and index the time frame by $a\in\mathbb{Z}^{+}$ in which $\mathcal{F}(a)=[(a-1)T_0+1,\cdots, aT_0]$.
Furthermore, $h_{ij}^{(s)}\in\mathcal{H}_{ij}^{(s)}$ denotes the channel gain (including path loss and channel fading) from transmitter  $i$ to receiver $j$ over sub-carrier $s$ where $\mathcal{H}_{ij}^{(s)}$ is a finite set.
Note that all channels are independent and experience block fading over time slots.

At the beginning of frame $a$, BS $b$ uploads its local information to the controller via the fronthaul by equally allocating the power budget $|\mathcal{S}|\totalpower{b}$ over all sub-carriers. For a given target  fronthaul rate $R^{\mathrm{U}}$,  the required time $\tau_{b}^{\rm U}(a)$ to upload the information at BS $b\in\mathcal{B}$ satisfies,
\begin{equation}\label{Eq: FH time cost}
\textstyle R^{\mathrm{U}}=  \sum_{s\in\mathcal{S}} \tau_{b}^{\rm U}(a)
\log_2\Big(1+\frac{\totalpower{b}h^{(s)}_{b{\rm C}}(a)}{N_0+\sum_{b'\in\mathcal{B}\setminus b}\totalpower{b'}h^{(s)}_{b'{\rm C}}(a)}\Big),
\end{equation}
where subscript ${\rm C}$ refers to the controller, and $N_0$ is the noise variance.
To gather information from all BSs, the controller requires $\max_{b\in\mathcal{B}}\{\tau_{b}^{\rm U}(a)\}$ time slots.
After manipulating the acquired local information, the controller sends back recommendations to every BS $b$ with a target rate  $R^{\rm D}$. Thus, the time required for feedback $\tau_{b}^{\rm D}(a)$ for all BSs $b\in\mathcal{B}$ satisfies,
\begin{equation}
\textstyle R^{\rm D}=\sum_{s\in\mathcal{S}} \tau_{b}^{\rm D}(a)
\log_2\Big(1+\frac{\totalpower  {\rm C} h^{(s)}_{{\rm C}b}(a)}{|\mathcal{B}|N_0+(|\mathcal{B}|-1)\totalpower {\rm C}h^{(s)}_{{\rm C}b}(a)}\Big).\label{Eq: FH download cost}
\end{equation}
In \eqref{Eq: FH download cost}, the controller equally allocates the available power $|\mathcal{S}|\totalpower {\rm C}$ to all BSs over all sub-carriers.
Finally, aided by the controller's suggestions, BS $b$ serves its associated MUs in the remaining slots of frame $a$. It is assumed that  sub-carriers are shared in both  DL and fronthaul. To avoid  interference between the DL and fronthaul, BSs wait for  $\tau(a)=\max_{b\in\mathcal{B}}\{\tau_{b}^{\rm U}(a)\}+\max_{b\in\mathcal{B}}\{\tau_{b}^{\rm D}(a)\}$ time slots, i.e., round trip time for the controller's recommendations. Thus, the overhead of the fronthaul is $\tau(a)$ which belongs to the finite set $\mathcal{T}$, and  $T_0-\tau(a)$ time slots are left for DL transmission within frame $a$.
It is assumed that the knowledge of $\tau$ and the perfect channel information of the signaling links are available at both the central controller and BSs. Moreover, BSs have full information of the direct links to the associated MUs.

In the DL transmission, BS $b$ serves MU $m$ with power $P^{(s)}_{bm}$ over sub-carrier $s$ in which $P^{(s)}_{bm}\in\mathcal{L}=\{0,\totalpower{b},\cdots,|\mathcal{S}|\totalpower{b}\}$ and  $\sum_{m\in\mathcal{M}_b,s\in\mathcal{S}} P_{bm}^{(s)}\leq |\mathcal{S}|\totalpower{b}$. Each sub-carrier is used by the BS to serve one MU, but can be reused in other cells resulting in inter-cell interference. Considering the available time in the DL, i.e., $T_0-\tau(a)$, the effective DL rate of MU $m$ served by BS $b$ over sub-carrier $s$ and time period $t\in\mathcal{F}(a)$ is given by,
\begin{multline}
\hspace{-0.9em}\textstyle R_{bm}^{(s)}(t)  =  \frac{T_0-\tau(a)}{T_0}
 \log_2\bigg(1  + \frac{P_{bm}^{(s)}(t)h_{bm}^{(s)}(t)}{N_0+\sum\limits_{b'\in\mathcal{B}\setminus b \atop m'\in\mathcal{M}_{b'}}P_{b'm'}^{(s)} (t) h_{b'm}^{(s)}(t)} \bigg).\label{Eq: DL instantaneous rate}\!\!
\end{multline}
It can be noted that $R_{bm}^{(s)}$ is bounded above by a maximal rate $R_{\rm max}$ due to the fact that the channel gains and transmit power are upper bounded. Moreover, each BS has queue buffers to store the traffic arrivals for MUs from the core network.  Let $Q_{bm}(t)$ denote the data queue at the beginning of slot $t$  for MU $m\in\mathcal{M}_b$ which  evolves as follows:
\begin{equation}
\textstyle Q_{bm}(t+1)=\max\Big\{Q_{bm}(t) - \sum\limits_{s\in\mathcal{S}}R_{bm}^{(s)}(t),0\Big\}+A_{bm}(t).\label{Eq: Queue-Q}
\end{equation}
Here, $A_{bm}(t)$ is the data arrival for MU $m\in\mathcal{M}_b$ during time slot $t$ which is independent and identically distributed over time with the mean $\lambda_{bm}>0$ and bounded above by a finite value $A_{{\rm max}}$, i.e., $0\leq A_{bm}(t)\leq A_{{\rm max}}$.
Since it is infeasible to store data indefinitely, BSs need to ensure that queue buffers are \emph{mean rate stable}, i.e., {$\lim\limits_{t\to\infty}{\mathbb{E}\left[|Q_{bm}(t)|\right]}/{t}= 0$} for all $b\in\mathcal{B},m\in\mathcal{M}_b$.
%
%
%
%
%
%
%
Note that satisfying queue stability is equivalent to ensuring that the average serving rate is larger than the mean arrival \cite{Neely/Stochastic}. Therefore, the stability condition becomes,
%
%
%
$\sum_{s\in\mathcal{S}}\bar{R}^{(s)}_{bm}=\lim\limits_{T\to\infty}\sum_{t=1}^{T}\sum_{s\in\mathcal{S}}\mathbb{E}\big[R^{(s)}_{bm}(t)\big]/T\geq \lambda_{bm}$.
%
%
%
Thus, as BSs compete one another over the limited resources to maximize the average DL rate, they need to satisfy,
\begin{align}
\textstyle \sum\limits_{m\in\mathcal{M}_b,s\in\mathcal{S}}\bar{R}^{(s)}_{bm}\geq \lambda_b= \sum\limits_{m\in\mathcal{M}_b}\lambda_{bm},~\forall\,b\in\mathcal{B}. \label{Eq: BS QoS}
\end{align}

\section{Stochastic Game among BSs}\label{Sec: Stochastic Game}

The competition among BSs over resources to maximize their utilities (in terms of rates) under queue and channel uncertainties  is modeled as a stochastic game
 $\mathcal{G}\coloneqq\big(\mathcal{B},\mathcal{W},\mathcal{P},\{u_{b}\}_{b\in\mathcal{B}}\big)$ among $|\mathcal{B}|$ players/BSs.
Here, $\mathcal{W}$ and $\mathcal{P}$ are the state and action spaces of all players,
while $u_{b}:\mathcal{W}\times\mathcal{P}\to\mathbb{R}^{+}\cup \{0\}$ denotes player $b$'s utility. In each time slot $t$, BS $b$ observes a random state realization $\boldsymbol{\omega}_b(t)\coloneqq (\tau(a),h_{bm}^{(s)}(t),~m\in\mathcal{M}_b,s\in\mathcal{S})$ from the state space $\mathcal{W}_b$.  Then, BS $b$ chooses a transmit power vector, {$\mathbf{P}_b(t)\coloneqq\big (  P^{(s)}_{bm}(t),m\in\mathcal{M}_b,s\in\mathcal{S}  \big ) \in\mathcal{P}_b$} from its action space {$\mathcal{P}_b\coloneqq\big\{\mathbf{P}_{b}\big|\,P_{bm}^{(s)}\in\mathcal{L},\sum_{m\in\mathcal{M}_b}\mathbb{I}\big\{ P_{bm}^{(s)}>0\big\}\leq 1,\sum_{m\in\mathcal{M}_b,s\in\mathcal{S}} P_{bm}^{(s)}\leq |\mathcal{S}|\totalpower{b}\big\}$} to serve its MUs.
Here, $\mathbb{I}\{\cdot\}$ is the indicator function.
%
We define BS $b$'s utility $u_{b}(\boldsymbol{\omega}(t),\mathbf{P}(t))$ as the expected DL rate with respect to the interference channels at the associated MUs which is given by,
\begin{equation}\label{Eq: BS's utility}
\textstyle u_{b}(\boldsymbol{\omega}(t),\mathbf{P}(t))\coloneqq\sum\limits_{m\in\mathcal{M}_b,s\in\mathcal{S}}\mathbb{E}_{\boldsymbol{\iota}_b}\big[R_{bm}^{(s)}(t)\big],
\end{equation}
where $\boldsymbol{\iota}_b\coloneqq\{h_{b'm}^{(s)},b'\in\mathcal{B}\backslash b,m\in\mathcal{M}_b,s\in\mathcal{S}\}$ is the set of interference channel gains of $\mathcal{M}_b$. Additionally, $\boldsymbol{\omega}\coloneqq[\boldsymbol{\omega}_1,\cdots,\boldsymbol{\omega}_{|\mathcal{B}|}]\in\mathcal{W}$ is the global random state with $\mathcal{W}\coloneqq\mathcal{W}_1\times\cdots\times\mathcal{W}_{|\mathcal{B}|}$, and $\mathbf{P}\coloneqq[\mathbf{P}_1,\cdots,\mathbf{P}_{|\mathcal{B}|}]\in\mathcal{P}$ is the global control action with $\mathcal{P}\coloneqq\mathcal{P}_1\times\cdots\times\mathcal{P}_{|\mathcal{B}|}$.
%
%
%
Furthermore, $\Pr(\mathbf{P}(t)|  \boldsymbol{\omega}(t))$ is the mixed strategy of the stochastic game, and we focus on the \emph{mixed stationary and Markovian} strategy $\Pr(\mathbf{P}|\boldsymbol{\omega})$ in this work.
Thus, the long-term time average expected utility of BS $b$ is given by,
\begin{equation}\label{Eq: BS's stationary utility}
\textstyle\bar{u}_{b}=\sum\limits_{\boldsymbol{\omega}\in\mathcal{W},\mathbf{P}\in\mathcal{P}} \Pr(\boldsymbol{\omega})\Pr(\mathbf{P}|  \boldsymbol{\omega})u_b(\boldsymbol{\omega},\mathbf{P}).
\end{equation}
For the purpose of allocating resources among BSs, the controller acts as a game coordinator. According to BSs' requirements \eqref{Eq: BS QoS}, the controller provides suggestions for BSs.
In order to incentivize BSs to follow the controller's suggestions, the controller finds the strategy satisfying the CCE constraints.  In this regard, we consider a more general case, i.e., the $\epsilon$-coarse correlated equilibrium ($\epsilon$-CCE) \cite{jnl:perlaza13}.
\begin{definition}\label{Def: epsilon-CCE}
The probability distribution $\Pr(\mathbf{P}|\boldsymbol{\omega})$  is an $\epsilon$-CCE if  it satisfies, for some $\epsilon\geq 0$,
\begin{align}
\hspace{-0.6em}\bar{u}_b(\tilde{\boldsymbol{\omega}}_b,\tilde{\mathbf{P}}_b)&\leq \Pr(\tilde{\boldsymbol{\omega}}_b)\theta_b( \tilde{\boldsymbol{\omega}}_b) ,\forall\,b\in\mathcal{B},\tilde{\boldsymbol{\omega}}_b\in\mathcal{W}_b, \tilde{\mathbf{P}}_b\in\mathcal{P}_b ,\label{Eq: CCE-1} \\
%
\textstyle\bar{u}_b&\geq
 \textstyle\sum_{ \boldsymbol{\omega}_b\in\mathcal{W}_b} \Pr( \boldsymbol{\omega}_b)\theta_b( \boldsymbol{\omega}_b)-\epsilon,\forall\,b\in\mathcal{B},\label{Eq: CCE-2}
\end{align}
where $\bar{u}_b(\tilde{\boldsymbol{\omega}}_b,\tilde{\mathbf{P}}_b)=\sum_{\boldsymbol{\omega}\in\mathcal{W},\mathbf{P}\in\mathcal{P}|\boldsymbol{\omega}_b=\tilde{\boldsymbol{\omega}}_b}  \Pr(\boldsymbol{\omega})\Pr(\mathbf{P}|  \boldsymbol{\omega})\allowbreak\times u_b(\boldsymbol{\omega},\tilde{\mathbf{P}}_b, \mathbf{P}_{-b})$ is the conditional expected utility on the state $\tilde{\boldsymbol{\omega}}_b$.
Here, $\theta_{b}( \boldsymbol{\omega}_{b})$ is the maximal utility  BS $b$ can achieve at  state $\boldsymbol{\omega}_{b}$ when deviating from the strategy $\Pr(\mathbf{P}|  \boldsymbol{\omega})$.
\end{definition}
%
%
%
%
%
%
%
%
%

\section{Controller-assisted  Resource Allocation and User Scheduling}\label{Sec: Two-timescale control}

To ensure an $\epsilon$-CCE for the stochastic game among BSs, the controller solves the network utility maximization problem which is given by, 
%
%
$$\begin{array}{ccl}
{\mbox{\bf OP:}}&\underset{\Pr(\mathbf{P}|\boldsymbol{\omega}),\theta_b( \boldsymbol{\omega}_b)}{\text{maximize}} &\phi(\bar{u}_{1},\cdots,\bar{u}_{|\mathcal{B}|})
%
\\ &\mbox{subject to}&\mbox{QoS requirement \eqref{Eq: BS QoS}},
%
\\&&\mbox{$\epsilon$-CCE constraints \eqref{Eq: CCE-1} and  \eqref{Eq: CCE-2}},
%
\\&&\textstyle\sum_{\mathbf{P}\in\mathcal{P}} \Pr(\mathbf{P}|  \boldsymbol{\omega})=1,~\forall\,\boldsymbol{\omega}\in\mathcal{W},
\\&&\textstyle\Pr(\mathbf{P}|  \boldsymbol{\omega})\geq 0,~\forall\,\boldsymbol{\omega}\in\mathcal{W}, \mathbf{P}\in\mathcal{P},
\end{array}$$
%
%
%
where
%
%
$\phi(\bar{u}_{1},\cdots,\bar{u}_{|\mathcal{B}|})=\sum_{b\in\mathcal{B}}\lambda_b\ln\big(1+\bar{u}_b\big)$
%
%
is the network utility. In {\bf OP}, $u_b$ and $R_{bm}^{(s)}$ depend on the statistics of $\boldsymbol{\iota}_b$, i.e., distributions of all interference channels, as per \eqref{Eq: DL instantaneous rate}, \eqref{Eq: BS's utility}, and \eqref{Eq: BS's stationary utility}. Since the MU has only the knowledge of the aggregate interference, estimating the distribution for each individual interference channel is  impractical.
To cope with this, an auxiliary utility $v_b(\boldsymbol{\omega},\mathbf{P}) $ for BS $b$ is introduced as follows:
\begin{multline}
\textstyle v_{b}(\boldsymbol{\omega}(t),\mathbf{P}(t))\coloneqq \sum\limits_{m\in\mathcal{M}_b,s\in\mathcal{S}}\frac{T_0-\tau(a)}{T_0} \\
\textstyle \times\log_2\bigg(1+\frac{P_{bm}^{(s)}(t)h_{bm}^{(s)}(t)}{N_0+\sum\limits_{b'\in\mathcal{B}\backslash b,m'\in\mathcal{M}_{b'}}P_{b'm'}^{(s)} (t)[\mathcal{H}_{b'm}^{(s)}]_{\rm max}} \bigg),\label{Eq: Remodeled utility}
\end{multline}
where $[\mathcal{H}_{b'm}^{(s)}]_ {\rm max}$ is the maximal element in $\mathcal{H}_{b'm}^{(s)}$.
\begin{proposition}\label{Prop: mean rate stability}
Let $\bar{v}_{b}$ be the average auxiliary utility of BS $b$. Then, the mean rate stability in  \eqref{Eq: BS QoS} is assured if
$\bar{v}_b\geq \lambda_{b},\forall\,b\in\mathcal{B}$,
 is held.
%
%
%
%
\end{proposition}
\begin{proof}
Please refer to Appendix \ref{Lem: mean rate stability}.
\end{proof}
\begin{proposition}\label{Prop: Epsilon-CCE}
The CCE strategy with respect to $v_{b}(\boldsymbol{\omega},\mathbf{P})$ is an $\epsilon$-CCE with respect to $u_{b}(\boldsymbol{\omega},\mathbf{P})$.
\end{proposition}
\begin{proof}
Please refer to Appendix \ref{Lem: Epsilon-CCE}.
\end{proof}

Since $v_b$ is the lower bound of $u_b$ as in \eqref{Eq: Jensen's inequality}, Propositions  \ref{Prop: mean rate stability} and \ref{Prop: Epsilon-CCE} state that finding a CCE with respect to $v_b$ maximizes the lower bound of the optimal solution of {\bf OP}.
However, to solve {\bf OP}, each BS requires the knowledge of the other BSs' actions (i.e., $\mathbf{P}_{-b}$) and their impact on $v_b(\boldsymbol{\omega},\mathbf{P})$.
Since the MU has the ability to measure only the aggregate interference, estimating the impact of $\mathbf{P}_{-b}$ on $v_b(\boldsymbol{\omega},\mathbf{P})$  is also impractical. To address this, we assume that the mapping $v_{b}:\mathcal{W}\times\mathcal{P}\to\mathbb{R}^{+}\cup \{0\},\forall\,b\in\mathcal{B}$,  is known at the controller in advance.
Additionally, to ensure $\bar{v}_b\geq \lambda_{b},\forall\,b\in\mathcal{B}$, and the CCE, the controller needs $\Pr(\boldsymbol{\omega})$ and all $\lambda_{b}$ which in practice are unknown  even at the BSs.
Nevertheless, as time evolves, all the statistics can be empirically estimated by BSs.
For notational simplicity, we denote the estimated state distribution and mean data arrival as  $\hat{\Pr}(\boldsymbol{\omega})$ and $\hat{\lambda}_{b}$
which are the uploaded information sent via fronthaul.
Thus, the controller solves the modified problem of {\bf OP} given by,
\\
{\bf RP:} \vspace{-10pt}
$$\begin{array}{cll}
\underset{\hat{\Pr}_v(\mathbf{P}|\boldsymbol{\omega}),\hat{\theta}_b( \boldsymbol{\omega}_b)}{\text{maximize}} & \textstyle\sum_{b}\hat{\lambda}_b\ln\big(1+\hat{v}_b\big)
\\ \mbox{subject to}& \hat{v}_b\geq\hat{\lambda}_b, &\forall\,b,
\\&\bar{v}_{b}(\tilde{\boldsymbol{\omega}}_b, \tilde{\mathbf{P}}_b)
\leq \textstyle\hat{\Pr}(\tilde{\boldsymbol{\omega}}_b)\hat{\theta}_b( \tilde{\boldsymbol{\omega}}_b) , &\forall\,b,\tilde{\boldsymbol{\omega}}_b, \tilde{\mathbf{P}}_b,
\\& \hat{v}_b\geq\sum_{\boldsymbol{\omega}_b} \hat{\Pr}(\boldsymbol{\omega}_b)\hat{\theta}_b( \boldsymbol{\omega}_b), &\forall\,b,
%
\\& \textstyle\sum_{\mathbf{P}} \hat{\Pr}_v(\mathbf{P}|  \boldsymbol{\omega})=1, &\forall\,\boldsymbol{\omega},
\\& \textstyle\hat{\Pr}_v(\mathbf{P}|  \boldsymbol{\omega})\geq 0, &\forall\,\boldsymbol{\omega}, \mathbf{P},
\end{array}$$
%
where $\hat{\Pr}_v(\mathbf{P}|  \boldsymbol{\omega})$ is the strategy related to the auxiliary utility $v_b$ and the estimated statistics  {$\hat{\Pr}(\boldsymbol{\omega})=\hat{\Pr}(\tau)\prod_{b\in\mathcal{B},m\in\mathcal{M}_b,s\in\mathcal{S}}\hat{\Pr}(h_{bm}^{(s)})$}. Accordingly, we have {$\hat{v}_b=\sum_{\boldsymbol{\omega}\in\mathcal{W},\mathbf{P}\in\mathcal{P}}  \hat{\Pr}(\boldsymbol{\omega})\hat{\Pr}_v(\mathbf{P}|  \boldsymbol{\omega})v_{b}(\boldsymbol{\omega},\mathbf{P})$}.

Due to the fact that the objective and all the constraints in {\bf RP} are concave and affine functions, respectively, {\bf RP} is a convex optimization problem.
Owing to the fact that  analytically solving {\bf RP}
  is not tractable, we resort to CVX, a package for specifying and solving convex programs \cite{cvx}, by numerically finding the optimal strategy  $\hat{\Pr}^{*}_v(\mathbf{P}|\boldsymbol{\omega})$ of {\bf RP}.
Based on the strategy $\hat{\Pr}^{*}_v(\mathbf{P}|\boldsymbol{\omega})$,  for each global state $\boldsymbol{\omega}\in\mathcal{W}$, the controller generates $T_0$ global action realizations  $\check{\mathbf{P}}((a-1)T_0+1),\cdots,\check{\mathbf{P}}(aT_0)$. From all global states and the corresponding global action realizations $\check{\mathbf{P}}(t)$ for time slot $t$, there exists a mapping rule $\Psi^{*}(t):\mathcal{W}\to\mathcal{P}$ describing the above state-to-action relation, i.e., $\Psi^{*}(\boldsymbol{\omega};t)=\check{\mathbf{P}}(t)$.
Afterwards, the controller feeds $\{\Psi^{*}(t),t\in\mathcal{F}(a)\}$ back to BSs.

With the knowledge of $\Psi^{*}(t)$ and the observed $\boldsymbol{\omega}_b(t)$, BS $b$ obtains the suggested transmit action $\check{\mathbf{P}}_b(t)$ \cite{Neely:13:Arxiv}. Although following the suggested $\check{\mathbf{P}}_b(t)$ assures data throughput higher than the mean data arrival, it may suggest to schedule MUs considering only the channel quality neglecting their queues. Therefore, this transmission suggestion may result in high latency. To remedy to this issue, instead of following the suggested action, BS $b$ optimizes the transmit power over the recommended set of sub-carriers  $\mathcal{X}_b(t)=\{s\in\mathcal{S}|\sum_{m\in\mathcal{M}_b}\check{P}_{bm}^{(s)}(t)> 0\}$ and schedules its MUs.
Thus, utilizing the available sub-carriers, BS $b$ maximizes the average DL rate by scheduling the MUs ensuring their queue stability in the remaining time slots of frame $a$. This is modeled as the following sub-problem:
%
$$\begin{array}{ccl}
{\mbox{\bf SP:}}&\underset{\mathbf{P}_b(t)\in\mathcal{P}_b}{\text{maximize}} & \textstyle\sum\limits_{m\in\mathcal{M},s\in\mathcal{S}}\bar{R}_{bm}^{(s)}
\\ &\mbox{subject to}& \textstyle\lim\limits_{t\to\infty}\frac{\mathbb{E}\left[|Q_{bm}(t)|\right]}{t}\to 0, ~\forall\,m\in\mathcal{M}_b,
\\&& P_{bm}^{(s)}(t)=0, ~\forall\,t,m\in\mathcal{M}_b,s\notin\mathcal{X}_b(t),
\end{array}$$
%
%
%
%
which  can be solved using dynamic programming. However, dynamic programming suffers from high computational complexity. To cope with this complexity issue, we resort to the tools in the Lyapunov optimization framework.
Let $\mathbf{\Xi}_b(t)\coloneqq\{Q_{bm}(t),m\in\mathcal{M}_b\}$ be the combined queue at BS $b$.  The conditional Lyapunov drift-plus-penalty for slot $t$ is given by,%
\begin{equation}\label{Eq: Conditional Lyapunov drift}
\textstyle\mathbb{E}\Big[L(\mathbf{\Xi}_b(t+1))-L(\mathbf{\Xi}_b(t))-\sum\limits_{m\in\mathcal{M}_b,  s\in\mathcal{S}}VR_{bm}^{(s)}(t)
\Big|\mathbf{\Xi}(t)\Big]
\end{equation}
with a parameter $V\geq 0$ which decides the trade-off between optimality and queue stability, and $L(\mathbf{\Xi}_b(t))=\frac{1}{2}\sum_{m\in\mathcal{M}_b}\big(Q_{bm}(t)\big)^2$ is the Lyapunov function.
Using $\max\{x,0\}\leq x^2$,  \eqref{Eq: Conditional Lyapunov drift} can be simplified as follows,
\begin{align}
&\eqref{Eq: Conditional Lyapunov drift}
%
%
%
%
\leq\textstyle \frac{|\mathcal{M}_b||\mathcal{S}|^2R^2_{{\rm max}}}{2}+\frac{|\mathcal{M}_b|A_{{\rm max}}^2}{2}+\mathbb{E}\bigg[\sum\limits_{m\in\mathcal{M}_b}Q_{bm}(t) \notag
\\&~\textstyle\times \Big(A_{bm}(t)-\sum\limits_{s\in\mathcal{S}}R_{bm}^{(s)}(t)\Big)-\sum\limits_{m\in\mathcal{M}_b\atop s\in\mathcal{S}}VR_{bm}^{(s)}(t)
\Big|\mathbf{\Xi}(t)\bigg].\label{Eq: Lyapunov bound}
\end{align}
The solution to {\bf SP} can be found by minimizing the upper bound in \eqref{Eq: Lyapunov bound} in every time slot \cite{Neely/Stochastic}. Thus, relaxing the power constraint in {\bf SP}, the MU scheduling problem during time slot $t$ at BS $b$ becomes,
%
%
%
$$\begin{array}{ccl}
{\mbox{\bf BP:}}&\underset{P_{bm}^{(s)}}{\text{maximize}} & \sum\limits_{m\in\mathcal{M}_b,s\in\mathcal{S}}(Q_{bm}+V)
\\ && \quad \times\mathbb{E}_{I_{bm}^{(s)}}  \bigg[\ln\bigg(1+\frac{P_{bm}^{(s)}h_{bm}^{(s)}}{N_0+ I_{bm}^{(s)}} \bigg)\bigg|\boldsymbol{\omega}_b,\mathcal{X}_b\bigg]
\\& \mbox{subject to}& \sum\limits_{m\in\mathcal{M}_b,s\in\mathcal{S}}P_{bm}^{(s)} \leq |\mathcal{S}|\totalpower{b},
\\&&P_{bm}^{(s)}  \geq 0, ~ \forall\,m\in\mathcal{M}_b,s\in\mathcal{X}_b,
\\&&P_{bm}^{(s)}  = 0, ~ \forall\,m\in\mathcal{M}_b,s\notin\mathcal{X}_b.
\end{array}$$
%
%
%
Here, the expectation is over the conditional distribution  of the estimated aggregate interference in slot $t$, denoted by $\hat{\Pr}\big(I_{bm}^{(s)};t|\boldsymbol{\omega}_b,\mathcal{X}_b\big)$.
For simplicity, the time index $t$ is omitted in the interference term $I_{bm}^{(s)}=\sum_{b'\in\mathcal{B}\backslash b,m'\in\mathcal{M}_{b'}}P_{b'm'}^{(s)}  h_{b'm}^{(s)}$.
%
%
%
Applying the KKT conditions,
the feasible power allocation $P_{bm}^{(s)*}$, $\forall\,s\in\mathcal{X}_b$, is found by
\begin{equation}
\mathbb{E}_{I_{bm}^{(s)}}\bigg[\frac{\big(Q_{bm}+V\big)h_{bm}^{(s)}}{N_0+ I_{bm}^{(s)}+P_{bm}^{(s)*}h_{bm}^{(s)}} \bigg|\boldsymbol{\omega}_b,\mathcal{X}_b\bigg] =\gamma_b\label{Eq: BS water-willing}
\end{equation}
%
%
%
%
%
if $\mathbb{E}_{I_{bm}^{(s)}}\Big[\frac{(Q_{bm}+V)h_{bm}^{(s)}}{N_0+ I_{bm}^{(s)}}\Big|\boldsymbol{\omega}_b,\mathcal{X}_b \Big]>\gamma_b$, and
$P_{bm}^{(s)*}=0$ otherwise.
Here, $\gamma_b$ is chosen to satisfy
$\sum_{m\in\mathcal{M}_b,s\in\mathcal{S}}P_{bm}^{(s)*} =|\mathcal{S}|\totalpower{b}$.
Finally, we search for the nearest (with respect to the Euclidean distance) $\mathbf{P}_b^{*}$ in $\mathcal{P}_b $ as the transmit power.
\begin{remark}\label{Remark: V impact}
From \eqref{Eq: BS water-willing}, it can be observed that for small $V$, BS $b$ has high priority to serve  MU $m$ with a large queue $Q_{bm}$.
Analogously, for large $V$, the BS allocates higher power to the MUs with better links.
As  such, the BS maximizes its DL rate while allowing the queues of the MUs with worse links to grow.
\end{remark}
By the end of DL transmission in time slot $t$,  MUs send feedback of the aggregate interference to the BS. Then, the BS updates $Q_{bm}(t+1)$ according to \eqref{Eq: Queue-Q} and the estimated interference statistics  as follows:
\begin{multline}\label{Eq: Empirical interference distribution-1}
\textstyle \hat{\Pr}\big(\tilde{I}_{bm}^{(s)};t+1|\tilde{\boldsymbol{\omega}}_b,\tilde{\mathcal{X}}_b\big)
=
 \frac{\sum_{\xi=1}^{t}\mathbb{I}\{\boldsymbol{\omega}_b(\xi)=\tilde{\boldsymbol{\omega}}_b,\mathcal{X}_b(\xi)=\tilde{\mathcal{X}}_b\}}{1+\sum_{\xi=1}^{t}\mathbb{I}\{\boldsymbol{\omega}_b(\xi)=\tilde{\boldsymbol{\omega}}_b,\mathcal{X}_b(\xi)=\tilde{\mathcal{X}}_b\}}
 \\\textstyle\times\hat{\Pr}\big(\tilde{I}_{bm}^{(s)};t|\tilde{\boldsymbol{\omega}}_b,\tilde{\mathcal{X}}_b\big)
+\frac{\mathbb{I}\big\{I_{bm}^{(s)}(t)=\tilde{I}_{bm}^{(s)}\big\} }{1+\sum_{\xi=1}^{t}\mathbb{I}\{\boldsymbol{\omega}_b(\xi)=\tilde{\boldsymbol{\omega}}_b,\mathcal{X}_b(\xi)=\tilde{\mathcal{X}}_b\}}
\end{multline}
if $\boldsymbol{\omega}_b(t)=\tilde{\boldsymbol{\omega}}_b$ and $\mathcal{X}_b(t)=\tilde{\mathcal{X}}_b$.
 Otherwise,
\begin{align}\label{Eq: Empirical interference distribution-2}
\textstyle \hat{\Pr}\big(\tilde{I}_{bm}^{(s)};t+1|\tilde{\boldsymbol{\omega}}_b,\tilde{\mathcal{X}}_b\big)
= \hat{\Pr}\big(\tilde{I}_{bm}^{(s)};t|\tilde{\boldsymbol{\omega}}_b,\tilde{\mathcal{X}}_b\big).
\end{align}
Furthermore, the BS empirically updates the estimated state distributions of ${h}_{bm}^{(s)}$ and $\tau$, and the mean arrival ${\lambda}_{b}$ for the purpose of upload to the controller during frame $(a+1)$.
%
%
%

The overheads in fronthaul depend on the amount of information which needs to be uploaded.
Based on the cardinalities of sets of channels and waiting time along the mean arrivals, each BS uploads $\sum_{m\in\mathcal{M}_b,s\in\mathcal{S}}|\mathcal{H}_{bm}^{(s)}|+|\mathcal{T}|+1$ amount of statistical values via the fronthaul.
%
Thus, $R^{\rm U}$ in \eqref{Eq: FH time cost} becomes  $(\sum_{m\in\mathcal{M}_b,s\in\mathcal{S}}|\mathcal{H}_{bm}^{(s)}|+|\mathcal{T}|+1)R_{\rm unit}$, where $R_{\rm unit}$ is the required transmission rate to upload a single statistical value.
Moreover, assuming that all possible mapping rules $\Psi$ are available at both the controller and BSs, the controller  sends $T_0$ real values representing the mappings $\{\Psi^{*}(t),t\in\mathcal{F}(a)\}$. Thus, $R^{\rm D}$ in \eqref{Eq: FH download cost} becomes $T_0R_{\rm unit}$.
\begin{figure}[!t]
\vspace{1em}
\centering
	{\def\svgwidth{\columnwidth}
		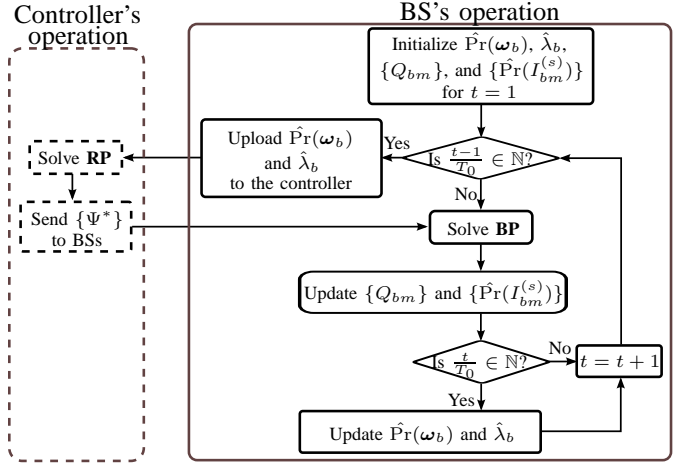}
			\vspace{-2em}
	\caption{Information flow digram of the two-timescale software-defined control approach.}
	\label{Fig: Flow diagram} 	
	\vspace{-2em}
\end{figure}
\begin{remark}
The proposed  software-defined approach operates in two timescales. In the long timescale,  BS $b$ uploads $\hat{\Pr}(\boldsymbol{\omega}_b)$ and $\hat{\lambda}_b$ at each time frame. Based on the uploaded information, the controller finds the set of optimal CCE mapping rules $\{\Psi^{*}(t)\}$ and feeds them back to the BSs. In the short timescale, i.e., at every time slot, BS $b$ schedules its MUs by solving {\bf BP}.
\end{remark}
The information flow diagram of the  two-timescale software-defined control approach is shown in Fig.~\ref{Fig: Flow diagram}.

\section{Numerical Results}\label{Sec: Numerical Results}
%
%
%
%
%
%
In this section, we validate the proposed solution for a network composed of  two small cell BSs and two MUs per cell.
For notation clarity, we refer to the BS and the two MUs in the first cell as BS 1, MU 1, and MU 2 while for the second cell, BS 2, MU 3, and MU 4, respectively.
We consider an indoor communication environment with the path loss model $30\log d+20\log 2.4+46$, where $d$ in meters is the distance between the transmitter and receiver \cite{rpt:itu_indoor}.
At MUs 1 and 3, the distances to the associated and interfering BSs are 10\,m and 40\,m while at MUs 2 and 4, they are 20\,m and 30\,m, respectively.
Further,  all wireless channels experience Rayleigh fading with unit variance, and channel gains are quantized into two levels.  The time cost $\tau$ is quantized into two levels with $\mathcal{T}=\{0.25,0.5\}$.
Additionally, we assume that the fronthaul SNR measured at the controller is 20\,dB. Coherence time and the bandwidth in each sub-channel  are 100\,ms and 10\,MHz, respectively.
We consider Poisson data arrival processes with  mean values $\lambda_{11}=\lambda_{12}=8$\,Mbps and $\lambda_{23}=\lambda_{24}=5$\,Mbps.
Moreover, $\totalpower 1=\totalpower 2=20$\,dBm, $\totalpower {\rm C}=25$\,dBm, $N_0=-85$\,dBm, $|\mathcal{S}|=2$, $T_0=10$, and $R_{\rm unit}=0.25\log_21.05$\,(bit/s/Hz). For performance comparison, we consider a non-SDN scheme without the central controller, and BSs do not communicate with each other. In this baseline, the BS solves {\bf BP} without the available sub-carrier  constraint, i.e., $\mathcal{X}_b(t)=\mathcal{S},\forall\,t,b\in\mathcal{B}$, and $\tau=0$.
%
%
%
%

The impact of the trade-off parameter $V$ on the BS's average DL rate and queue length for the proposed and baseline schemes is illustrated in Fig.~\ref{Fig: 2} and  Fig.~\ref{Fig: 3}, respectively. Note that the average queue length is proportional to the average waiting time of a packet in the queue (i.e., $\bar{Q}/\lambda$) and thus, represents the average delay/latency. For small $V $, Lyapunov optimization aims at minimizing the average queues while for large $V$, the focus is on maximizing the average rates. Thus, the lowest queues can be observed at $V=0$  while the highest rates can be seen at $V=100$ for both SDN and non-SDN schemes. Furthermore, Fig.~\ref{Fig: 2} illustrates that the proposed approach is aware of the traffic demands in which higher rates are ensured for BS 1 over BS 2 while the non-SDN scheme provides equal rates for both BSs. According to Fig.~\ref{Fig: 3}, it can be noted that the SDN scheme maintains a lower queue length for the BS with high demands (BS 1) over the non-SDN scheme, i.e. a significant delay reduction for high traffic demands is ensured. Although there is no improvement in the latency of BS 2, the proposed scheme provides higher rates compared to the non-SDN scheme.

%
%

%
%


\begin{figure}[!t]
	\centering
	\vspace{-1.3em}
	\includegraphics[width=0.9\columnwidth]{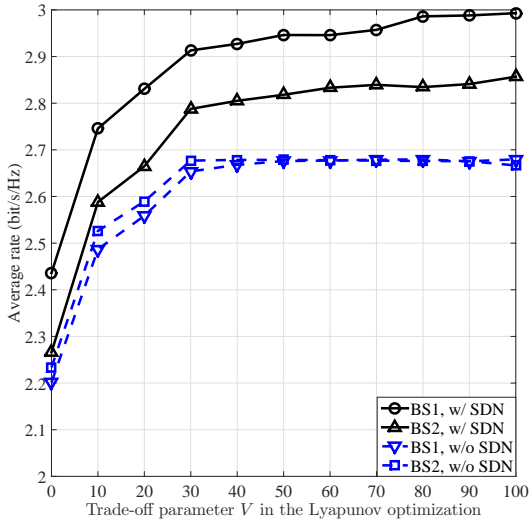}	
	\vspace{-1.5em}
	\caption{BS's average DL rate as $V$ varies.}
\vspace{-1.5em}
	\label{Fig: 2}
\end{figure}
\begin{figure}[!t]
	\centering
	\vspace{-1.3em}
	\includegraphics[width=0.9\columnwidth]{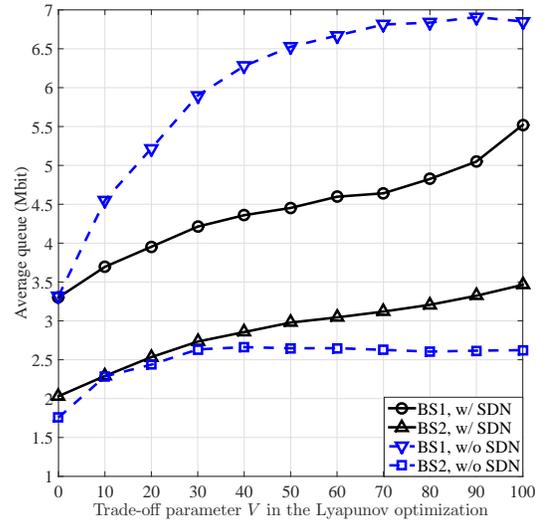}	
	\vspace{-1.5em}
	\caption{BS's average queue  as $V$ varies.}
	\label{Fig: 3}
\end{figure}
%
%
%
%
%
%
%
\begin{figure}[!h]
	\centering
	\vspace{-1.3em}
	\includegraphics[width=0.9\columnwidth]{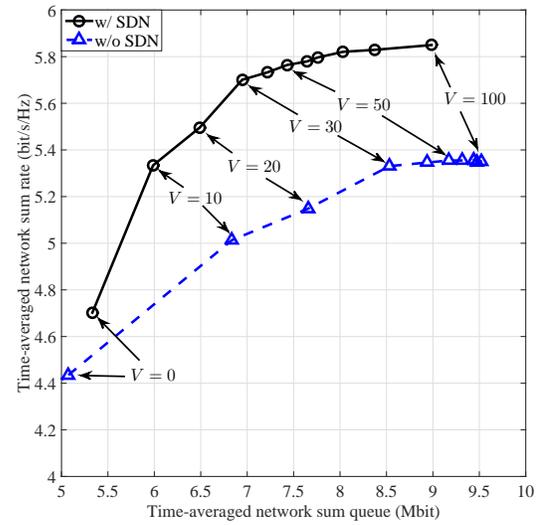}
	\vspace{-1.5em}
\caption{Trade-off between the time-averaged sum rate and queue size.}
	\label{Fig: R_Q_tradeoff}
		\vspace{-1.5em}
\end{figure}
%
%
%
The trade-off between the time-averaged network sum rate and the aggregate queue is illustrated in Fig.~\ref{Fig: R_Q_tradeoff}. It can be noted that for a given target rate, the proposed SDN approach yields lower latency compared to the non-SDN scheme. Moreover, Fig.~\ref{Fig: R_Q_tradeoff} shows that increasing $V$ improves the rates with a price of the latency which grows as $V$ increases. Therefore, $V=50$ can be selected as the preferred trade-off parameter due to the fact that the cost of latency dominates the gains in rates for values of $V>50$, for both schemes. For the choice of $V=50$, the proposed SDN scheme yields about 8\% gains in rates and 19\% reductions in latency over the non-SDN scheme.

%
%
%

%
%

Finally, the impact of the fronthaul SNR on the throughput and latency on the proposed two-timescale SDN approach is shown in Fig.~\ref{Fig: SNR_Q}. Note that the non-SDN scheme has no impact from the reliability of the fronthaul.
For a low-quality fronthaul, the round trip time for the controller's recommendations, $\tau$, is larger than $[\mathcal{T}]_{\max}$, i.e., the maximum element in $\mathcal{T}$. Therefore, BSs do not receive any recommendation even after waiting $[\mathcal{T}]_{\max}$  time slots and thus, cannot utilize their DL transmission resulting lower rates and higher delay compared to the non-SDN scheme. As the fronthaul becomes reliable, BSs obtain the recommendations with $\tau\leq [\mathcal{T}]_{\max}$. Thus, a significant improvement in throughputs can be observed. These gains in rates dominate the overhead in fronthaul and thus, improve the overall latency.
\begin{figure}[!t]
	\centering
	\vspace{-1.3em}
	\includegraphics[width=0.9\columnwidth]{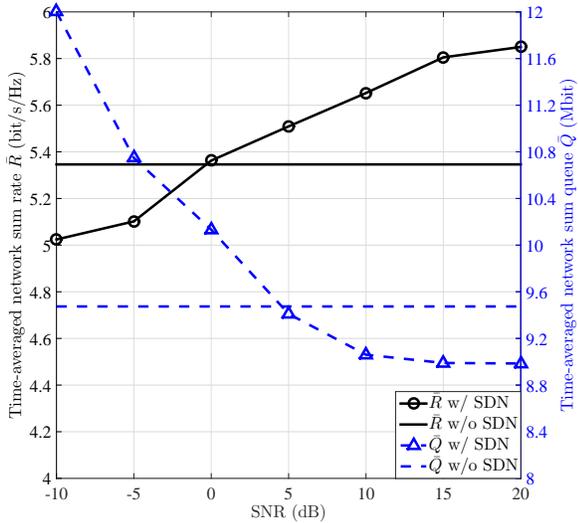}
	\vspace{-1.5em}
	\caption{Time-averaged sum rate and queue size as the fronthaul SNR varies, $V$ = 100.}
	\label{Fig: SNR_Q}
	\vspace{-1.5em}
\end{figure}

\section{Conclusions}\label{Sec: Conclusions}

In this paper, we propose a software-defined control mechanism for wireless networks with an in-band fronthaul. Here, BSs compete over the wireless resources in both the DL and fronthaul to maximize their utilities in terms of the average DL rates under the uncertainties of queues and channel states. Thus, the utility maximization problem is cast as a dynamic stochastic game among BSs where a central controller with a global view coordinates BSs via the fronthaul. Leveraging the timescale separation among the central controller and BSs, the controller ensures a CCE in the network by providing recommendations for BSs in the long timescale while BSs serve MUs in the short timescale. The tools of  Lyapunov optimization are invoked to provide a low-complexity traffic-aware user scheduling algorithm at each BS. Simulation results yield  throughput and latency enhancements over a non-SDN scheme. Furthermore, it is noted that the cost of overhead in fronthaul becomes negligible compared to the gains in throughput and  latency as the reliability of the fronthaul improves.

\section*{Acknowledgments}
This research was supported by TEKES grant 2364/31/2014 and the Academy of Finland  project CARMA.

\appendices
\section{Proof of Proposition \ref{Prop: mean rate stability}}
\label{Lem: mean rate stability}
As the aggregate interference in $v_b$ is modeled by the maximal interference channel gain $[\mathcal{H}_{b'm}^{(s)}]_{\rm max}$, $u_b$ is lower bounded by $v_b$, i.e.,
\begin{equation}\label{Eq: Jensen's inequality}
u_{b}(\boldsymbol{\omega},\mathbf{P})\geq v_{b}(\boldsymbol{\omega},\mathbf{P}),~\forall\,\boldsymbol{\omega}\in\mathcal{W},\mathbf{P}\in\mathcal{P}.
\end{equation}
Taking the expectation of \eqref{Eq: Jensen's inequality} with respect to $\Pr(  \boldsymbol{\omega})\Pr(\mathbf{P}|  \boldsymbol{\omega})$ and applying $\bar{v}_b\geq \lambda_{b}$,
we obtain
%
%
$ \sum_{m\in\mathcal{M}_b,s\in\mathcal{S}}\bar{R}_{bm}^{(s)}
%
%
\geq\lambda_{b}.$

%
%
%
\section{Proof of Proposition \ref{Prop: Epsilon-CCE}}
\label{Lem: Epsilon-CCE}
Let $\Pr_v(\mathbf{P}|  \boldsymbol{\omega})$ be the CCE with respect to $v_{b}$. Thus, from \eqref{Eq: CCE-1} and \eqref{Eq: CCE-2}, it satisfies,
\begin{align}
\bar{v}_{b}(\tilde{\boldsymbol{\omega}}_b,\tilde{\mathbf{P}}_{b}) &\leq \Pr(\tilde{\boldsymbol{\omega}}_b)\theta_b( \tilde{\boldsymbol{\omega}}_b), \label{Eq: Auxiliary CCE-1} \\
 \bar{v}_{b} &\geq \textstyle\sum_{\boldsymbol{\omega}_b} \Pr(\boldsymbol{\omega}_b)\theta_b( \boldsymbol{\omega}_b).\label{Eq: Auxiliary CCE-2}
\end{align}
Furthermore,  \eqref{Eq: Jensen's inequality} can be rewritten as follows:
\begin{equation}\label{Eq: Epsilon-CCE proof 1}
u_{b}(\boldsymbol{\omega},\mathbf{P})= v_{b}(\boldsymbol{\omega},\mathbf{P})+ \delta_{b}(\boldsymbol{\omega},\mathbf{P}),
\end{equation}
with $\delta_{b}(\boldsymbol{\omega},\mathbf{P})\geq 0$.
Applying \eqref{Eq: Epsilon-CCE proof 1} to \eqref{Eq: Auxiliary CCE-1}, it holds that,
\begin{align}
&\textstyle \sum\limits_{\boldsymbol{\omega}\in\mathcal{W}|\boldsymbol{\omega}_b=\tilde{\boldsymbol{\omega}}_{b}}
 \sum\limits_{\mathbf{P}\in\mathcal{P}}  \Pr(\boldsymbol{\omega})\Pr_v(\mathbf{P}|  \boldsymbol{\omega})u_{b}(\boldsymbol{\omega},\tilde{\mathbf{P}}_{b},\mathbf{P}_{-b}),\notag
%
%
%
\\&\qquad \leq\Pr(\tilde{\boldsymbol{\omega}}_b)\big(\theta_b( \tilde{\boldsymbol{\omega}}_b)+\epsilon_b( \tilde{\boldsymbol{\omega}}_b)\big)=\Pr(\tilde{\boldsymbol{\omega}}_b)\eta_b( \tilde{\boldsymbol{\omega}}_b),\label{Eq: Epsilon-CCE proof 2}
\end{align}
where
\begin{align*}
&\textstyle\epsilon_b( \tilde{\boldsymbol{\omega}}_b)=\max\Big\{0,\max\limits_{\tilde{\mathbf{P}}_b\in\mathcal{P}}\Big\{-\theta_b( \tilde{\boldsymbol{\omega}}_b)+\sum\limits_{\boldsymbol{\omega}\in\mathcal{W}|\boldsymbol{\omega}_b=\tilde{\boldsymbol{\omega}}_{b}}
 \sum\limits_{\mathbf{P}\in\mathcal{P}}
\\&\textstyle\quad\frac{\Pr(\boldsymbol{\omega}){\Pr}_v(\mathbf{P}|  \boldsymbol{\omega})}{{\Pr(\tilde{\boldsymbol{\omega}}_b)}} \big( v_b(\boldsymbol{\omega},\tilde{\mathbf{P}}_{b},\mathbf{P}_{-b})+\delta_b(\boldsymbol{\omega},\tilde{\mathbf{P}}_{b},\mathbf{P}_{-b})\big)\Big\}\Big\},
\end{align*}
and $\eta_b( \tilde{\boldsymbol{\omega}}_b) =\theta_b( \tilde{\boldsymbol{\omega}}_b)+\epsilon_b( \tilde{\boldsymbol{\omega}}_b)$ is an auxiliary variable.
By applying \eqref{Eq: Jensen's inequality}  and \eqref{Eq: Epsilon-CCE proof 2} to \eqref{Eq: Auxiliary CCE-2},
\begin{equation}
\bar{u}_{b}\textstyle\geq\sum_{\boldsymbol{\omega}_b\in\mathcal{W}_b}\Pr(\boldsymbol{\omega}_b)\eta_b( \boldsymbol{\omega}_b)-\epsilon,\label{Eq: Epsilon-CCE proof 3}
\end{equation}
where $\epsilon=\max_{b\in\mathcal{B}}\big\{\sum_{\boldsymbol{\omega}_b\in\mathcal{W}_b}\Pr(\boldsymbol{\omega}_b)\epsilon_b( \boldsymbol{\omega}_b)\big\}$. From  \eqref{Eq: Epsilon-CCE proof 2} and \eqref{Eq: Epsilon-CCE proof 3}, we note that $\Pr_v(\mathbf{P}|  \boldsymbol{\omega})$ is an $\epsilon$-CCE of $u_b(\boldsymbol{\omega},\mathbf{P})$.


\bibliographystyle{IEEEtran}
\bibliography{ref}

\end{document}